\documentclass[a4paper,11pt]{amsart}

\usepackage[latin1]{inputenc}
\usepackage{amsmath}
\usepackage{amsthm}
\usepackage{amscd}
\usepackage{amssymb}
\usepackage{amsfonts}
\usepackage{pb-diagram, pb-xy}
\usepackage[all]{xy}
\usepackage{hyperref}
\usepackage{MnSymbol}
\hypersetup{colorlinks=true}

\newcommand{\bbZ}{\mathbb{Z}}
\newcommand{\bbN}{\mathbb{N}}

\newcommand{\bbR}{\mathbb{R}}
\newcommand{\bbC}{\mathbb{C}}
\newcommand{\bbP}{\mathbb{P}}

\newcommand{\calC}{\mathcal{C}}
\newcommand{\calD}{\mathcal{D}}
\newcommand{\F}{\mathbb{F}}

\newcommand{\N} {\ensuremath{{\rm N}}}

\newcommand{\Tr} {\ensuremath{\textnormal{Tr}}}

\newcommand{\RM} {\mathcal{R}\mathcal{M}}
\newcommand{\PRM} {\mathcal{P}\mathcal{R}\mathcal{M}}
\newcommand{\GL} {{\rm GL}}
\newcommand{\PGL} {{\rm PGL}}
\DeclareMathOperator{\supp}{supp} \DeclareMathOperator{\id}{id}
\DeclareMathOperator{\Stab}{Stab} 
\DeclareMathOperator{\weight}{wt}

\newtheorem{thm1}{Theorem}
\newtheorem{que1}{Question}
\newtheorem*{thm*}{Theorem}
\newtheorem{thm}{Theorem}[section]

\newtheorem{pro}{Proposition}[section]
\newtheorem{lem}{Lemma}[section]

\newtheorem{cor}{Corollary}[section]

\newtheorem{que}{Question}[section]

\title{On the Stabilizer of Weight Enumerators of Linear Codes}
\author{Martino Borello, Olivier Mila}

\begin{document}

\begin{abstract}
This paper investigates the relation between linear codes and the stabilizer in ${\rm GL}_2(\bbC)$ of their weight enumerators. We prove a result on the finiteness of stabilizers and give a complete classification of linear codes with infinite stabilizer in the non-binary case. We present an efficient algorithm to compute explicitly the stabilizer of weight enumerators and we apply it to the family of Reed-Muller codes to show that some of their weight enumerators have trivial stabilizer.
\end{abstract}

\maketitle

\vspace{7mm}
\tableofcontents
\vspace{7mm}

\section{Introduction}\label{section-intro}

Counting $\mathbb{F}_q$-rational points on hypersurfaces is, in
general, a very difficult task. In \cite{El}, Noam Elkies uses
Coding Theory to give the point counts of cubic surfaces in the
$3$-dimensional projective space over $\mathbb{F}_q$. His idea was later
generalized by Nathan Kaplan in his PhD thesis \cite{Na}. The
starting point of their method is that a counting of points of
varieties belonging to a family $\mathcal{F}$ can be deduced by the
determination of the weight enumerator of a linear code associated
to $\mathcal{F}$. In fact, such code is related to well-known ones
in Coding Theory, namely the Reed-Muller codes (in their affine or
projective version).

One of the most remarkable theorems in Coding Theory is Andrew
Gleason's 1970 Theorem \cite{Gl}, which classifies the weight enumerators
of self-dual doubly-even codes.
The crucial argument in the proof this theorem is to observe that
self-duality and divisibility give interesting invariants for the
weight enumerator of self-dual codes:
if $p(x,y)\in \mathbb{Z}[x,y]$ is the homogeneous weight enumerator
of a self-dual doubly-even binary code, then
$$p(x,y)=p(x,iy)=p\left(\frac{x+y}{\sqrt{2}},\frac{x-y}{\sqrt{2}}\right).$$
This fact has a lot of significant implications on the shape of
$p(x,y)$ and consequently on the family of self-dual doubly-even
codes. Along those lines, we look for invariances of weight
enumerators of Reed-Muller codes. A divisibility condition is proved
by James Ax already in 1964 \cite{A} for the affine version of
Reed-Muller codes. The counterpart for the projective version can be
easily deduced. On the other hand, Reed Muller codes are not, in
general, self-dual. An analogue of the condition coming from
self-duality is much harder to be derived, since the nature of
weight enumerator is essentially combinatorial and not geometric, as
remarked in \cite{Ken94}.

We decided to approach the problem with a different strategy: since
some weight enumerators of Reed-Muller codes are known, we wanted to
calculate their invariants directly, with the hope to find general
invariants of a bigger family of codes. This happened to be not so easy in practice,
because the calculation of all invariants directly is
computationally complex. So we developed a new method based on the
action of ${\rm PGL}_{2}(\mathbb{C})$ on the projective line
$\mathbb{P}^1(\mathbb{C})$ (which is simply $3$-transitive). This
trick gives a huge restriction on potential elements in ${\rm
GL}_{2}(\mathbb{C})$.

Using this technique, we were able to prove the following.
\begin{thm1}\label{thm-1} The homogeneous weight enumerator of linear code has a finite number of
invariants if and only if its non-homogeneous version has at least
$3$ distinct complex roots.
\end{thm1}

We designed an algorithm to calculate approximated invariants, and this allowed us
to prove the existence of Reed Muller codes whose weight enumerator have trivial
stabilizer: for example, the Reed Muller codes which encodes conics
in the affine plane over $\mathbb{F}_4$.

\begin{thm1}
  The following codes have weight enumerator with trivial stabilizer in ${\rm PGL}_2(\bbC)$, i.e. with
  stabilizer in ${\rm GL}_2(\bbC)$ consisting only of scalar matrices:
  \[
  \RM_{4}(2,2), \quad \RM_{4}(3,2), \quad \RM_{5}(2,2),
  \]
  \[
  \PRM_{5}(3,2), \quad \PRM_{5}(3,2)^\perp = \PRM_{5}(5,2).
  \]
\end{thm1}

This suggests that no general non-trivial invariant can be found for
Reed Muller codes. However, three main questions arise from what we have
done.
\begin{que1}\label{Q1}
What are the codes whose weight enumerators have at most $2$ distinct
complex roots?
\end{que1}
\begin{que1}\label{Q2}
Can we find, with the algorithm described, new families of codes for
which there exist non-trivial invariants?
\end{que1}
\begin{que1}\label{Q3}
If the answer to Question \ref{Q2} is positive, can we determine
unknown weight enumerators of large codes?
\end{que1}

The present paper gives an almost complete answer to Question
\ref{Q1}. In particular, we prove the following classification
theorem.

\begin{thm1}
Let $\calC \subset \F_q^n$ be a linear code with weight enumerator
$W_\calC(x) \in \bbZ[x]$ having at most two distinct complex roots.
Then only the following possibilities may hold:
\begin{itemize}
\item[(a)] $W_\calC(x)=x^n$ and $\calC=\{\underline{0}\}$;
\item[(b)] $W_\calC(x)=(x+(q-1))^n$ and $\calC= \F_q^n$;
\item[(c)] $n$ is even, $W_\calC(x)=(x^2+(q-1))^{n/2}$ and, if $q\ne 2$,
$\calC\cong \bigoplus_{i=1}^{n/2}\langle(1,1)\rangle_{\F_q}$.
\end{itemize}
\end{thm1}

Question \ref{Q2} and \ref{Q3} are of course more general, and harder to answer.
However, while the approach fails for Reed-Muller codes, it does not imply that
the answer to these questions is negative.


In \S\ref{sec-back} we present the necessary background of Coding
Theory and we define invariants for weight enumerators. The
classification of linear codes with weight enumerator with at most
two distinct complex roots is given in \S\ref{sec-roots}. In
\S\ref{sec-proof1} we prove Theorem \ref{thm-1} about the finiteness
of the stabilizer. Finally, \S\ref{sec-algo} and \S\ref{sec-triv}
are devoted to the presentation of the algorithm and its
applications.


\section{Background}\label{sec-back}

A \emph{linear code} $\calC$ is a subspace of $\F_q^n$, where $n$ is
a positive integer called the \emph{length} of the code. A
\emph{generator matrix} of a linear code $\calC$ is a matrix whose
rows generates $\calC$. Elements of $\calC$ are called
\emph{codewords}. The support of a codeword $c\in\calC$, denoted
$\supp(c)$, is defined as follows:
$$\supp(c):=\{i\in \{1,\ldots,n\} \ | \ c_i\neq 0\}.$$
The \emph{weight} $\weight(c)$ of a codeword $c$ is the cardinality its support.
The (one-variable) \emph{weight enumerator} of $\calC \leq
\F_q^n$ is the polynomial
\[
W_\calC(x) :=  \sum_{c \in \calC} x^{n- \weight(c)} = \sum_{i=0}^n
a_i x^i, \quad a_i := \#\{\text{codewords of weight }n-i\}.
\]
We will call \emph{homogeneous weight enumerator} the homogeneous
version of $W_\calC(x)$, i.e.
\[
W_\calC(x,y) = \sum_{c \in \calC} x^{n- \weight(c)} y^{\weight(c)} =
\sum_{i=0}^n a_i x^i y^{n-i}.
\]
If $\calC\leq \F_q^m$ and $\calD\leq \F_q^n$ are linear codes with
generator matrices $C$ and $D$ respectively, their \emph{direct sum}
is the vector space $\calC \oplus \calD$ naturally embedded in
$\F_q^{m+n}$, i.e. the code with generator matrix
$$\left[\begin{array}{cc}C & 0 \\ 0 & D\end{array}\right].$$
Observe that
\[
W_{(\calC \oplus \calD)} = W_\calC \cdot W_\calD
\]
for the one-variable and homogeneous weight enumerators.

A \emph{monomial transformation} $f:\F_q^n \to \F_q^n$ is a linear
transformation of the form
\begin{align*}
  f: \F_q^n &\longrightarrow \F_q^n \\
  v & \longmapsto D P v
\end{align*}
where $D$ is an $n \times n$ diagonal matrix with non-zero diagonal
entries, and $P$ is an $n \times n$ permutation matrix. Two codes
are said to be \emph{equivalent} if one is the image of the other
under a monomial transformation. Observe that in the case $q=2$, a
monomial transformation is just a permutation. It is easy to observe
that two equivalent codes have the same weight enumerator.


\begin{par}{\bf Convention:}
Since we are interested in weight enumerators, we will usually identify
codes up to equivalence.
In particular, a generator matrix of a code $\calC$ will mean a generator matrix of some
code equivalent to $\calC$.
\end{par}


$\GL_2(\bbC)$ acts naturally on $\bbC[x,y]$ as follows:
\[
A\cdot p(x,y)=p(ax+by,cx+dy) \quad
\text{for } A=\left[\begin{smallmatrix}a&b\\c&d\end{smallmatrix}\right]\in
\GL_2(\bbC), \text{ and } p\in \bbC[x,y].
\]
Every element in the stabilizer
$\Stab_{\GL_2(\bbC)}(p)$ of $p$ is called \emph{invariant} of $p$.

There are two main properties which give rise to invariants of weight enumerators.

\begin{enumerate}
\item Let $\Delta > 1$ be an integer.
A linear code $\calC$ is \emph{divisible by $\Delta$} if the weight of every codeword of
$\calC$ is divisible by $\Delta$.
A code is \emph{divisible}, if it is divisible by some $\Delta>1$.
In terms of weight enumerator, a
linear code $\calC$ is divisible by $\Delta$ if and only if
$$D_{\Delta}:=\left[\begin{array}{cc}1&0\\0&\zeta_{\Delta}\end{array}\right]\in \Stab_{\GL_2(\bbC)}(W_\calC(x,y)),$$
where $\zeta_{\Delta}$ is a primitive $\Delta$-root of unity.

\item The \emph{dual} of a linear code $\calC \leq \F_q^n$ (denoted
$\calC^\perp$) is the orthogonal space with respect to the standard
inner product of $\F_q^n$, $\langle x, y \rangle = \sum_{i=1}^n x_i
y_i $ for $x, y \in \F_q^n$, i.e.
$$\calC^\perp:=\{v\in \F_q^n \ | \ \langle v,c\rangle=0 \ \text{for all }c\in \calC\}.$$
The relation between the weight enumerator of some code $\calC$ and its dual $\calC^\perp$ is given
by MacWilliams' Theorem.
\begin{thm*}[MacWilliams,  \cite{MS}]
  Let $\calC\leq \F_q^n$ be a linear code, and $\calC^\perp$ its dual. Then
  \[
  W_{\calC^\perp}(x,y) = \frac{1}{\#\calC} W_\calC(x + (q-1)y, x-y).
  \]
\end{thm*}
A linear code is \emph{self-dual} if $\calC^\perp = \calC$.
In that case we have $\#\calC = \#\calC^\perp = q^{n/2}$, and so MacWilliams' Theorem implies that
\[ S_q:=q^{-{1/2}}\left[
\begin{array}{cc}
  1 & q-1 \\
  1 & -1
\end{array}\right]\in \Stab_{\GL_2(\bbC)}(W_\calC(x,y)) \qquad (\star).
\]
A linear code which satisfy ($\star$) is called \emph{formally self
dual}.
\end{enumerate}

If $G\leq \GL_2(\bbC)$ is a group of matrices, the ring
\[\bbC[x,y]^G:=\{p(x,y)\in \bbC[x,y] \ | \ A\cdot p(x,y)=p(x,y) \ \text{for all }A\in G\}\]
is called the \emph{invariant ring} of $G$.

A beautiful result by Gleason \cite{Gl} states that the weight
enumerator of a self-dual binary linear code which is
\emph{doubly-even}, i.e. is divisible by $4$, lies in the ring
$$\bbC[x^8+14x^4y^4+y^8,x^4y^4(x^4-y^4)^4],$$
which is the invariant ring of $G:=\langle D_4,S_2\rangle$. This
result has a lot of consequences on the family $\mathcal{F}$ of
self-dual doubly-even binary linear codes. Among the others, the
length of a code in $\mathcal{F}$ is always divisible by $8$ and one
can derive upper bounds on the minimum non-zero weight of such
codes, depending only of the length.

As we said in the introduction, inspired by this result, we look for
invariances of weight enumerators of Reed-Muller codes. Let us
briefly introduce this family of codes.

The \emph{Reed-Muller code} $\mathcal{RM}_q(r,m)$ on $m$ variables,
of degree $r$ and defined over $\F_q$ is the code
\[
\mathcal{RM}_q(r,m):=\left\{(p(v))_{v\in \F_q^m}\ | \ p\in \F_q[x_1,\ldots,x_m]_r\right\},
\]
where $\F_q[x_1,\ldots,x_m]_r$ is the ring of polynomials in $m$
variables, with coefficient in $\F_q$ and of degree at most $r$.
Clearly, $\mathcal{RM}_q(r,m)$ encodes all the hypersurfaces in
$\mathbb{A}^m(\F_q)$ of degree at most $r$. So, determining the
weight enumerator of such a code is equivalent to counting
$\F_q$-rational points of hypersurfaces in the affine space.

Similarly, we can define the \emph{projective Reed-Muller code}
$\mathcal{PRM}_q(r,m)$ on $m$ variables, of degree $r$ and defined
over $\F_q$ in the following way:
$$\mathcal{PRM}_q(r,m):=\{(p(v))_{v\in R}\ | \ p\in \F_q[x_0,\ldots,x_m]^h_r\} \cup \{0\} ,$$
where $\F_q[x_1,\ldots,x_m]^h_r$ is the ring of degree $r$
homogeneous polynomials in $m+1$ variables, with coefficient in
$\F_q$, and $R$ is a set of representatives of all the points of
$\mathbb{P}^m(\F_q)$.
Observe that changing the set of representatives $R$ gives rise to an equivalent code.

Both Reed Muller codes and projective Reed Muller codes are
divisible codes, as a consequence of a theorem by Ax \cite{A}. In
general, they are not self-dual codes.


\section{Weight enumerators with at most two distinct complex
roots}\label{sec-roots}

One of the main problem in Coding Theory is to determine a code with
a prescribed weight enumerator. The problem in general is extremely
difficult, since, as remarked in \S \ref{section-intro}, codes are
geometric objects while weight enumerators are combinatorial
objects. In this section we will consider very particular weight
enumerators, which are important for our purposes, that is
polynomials with at most two distinct complex roots. Note that all
roots of weight enumerators are algebraic integers, since every
weight enumerator is a monic polynomial with integer coefficients.
We will give an almost complete classification of codes with such a
weight
enumerator.

We start with a quite technical lemma which is fundamental for our
classification.
\begin{lem} \label{lem:oddweight}
  Let $\calC$ be a code over $\F_q$ with $q \neq 2$. Assume all codewords of $\calC$ have even weight.
  Let $c \in \calC$ be a codeword of weight two, and $x \in \calC$ an arbitrary codeword.
  Let $\supp(c)=\{i,j\}$. Then there exists $\lambda\in \F_q$ such that $(x_i, x_j)=(\lambda c_i, \lambda c_j)$.
\end{lem}
\begin{proof}
 Suppose that $(x_i, x_j)\neq (\lambda c_i, \lambda c_j)$ for every
$\lambda\in
  \F_q$. In particular, $(x_i, x_j) \neq (0,0)$. Without lost of generality,
  we assume that $x_i \neq 0$.
  If $x_j = 0$, then there exists $\mu\in \F_q^\times$ such that
  $$x_i+\mu c_i\neq 0 \qquad \text{and} \qquad \mu c_j = x_j+\mu c_j\neq 0.$$
  Indeed, it suffices to take any $\mu\in \F_q^\times \setminus\{
  -x_ic_i^{-1}\}$, that is non-empty because $q>2$. Then
$$\supp(x + \mu c)=\supp(x) \ \cup \ \{j\}$$
  so that $x + \mu c$ has odd weight, which gives a contradiction.
  If $x_j \neq 0$, then
  $$\supp(x - (x_i c_i^{-1}) c)=\supp(x) \setminus \{i\},$$
so that $x - (x_i c_i^{-1}) c$ has odd weight, which gives again a
contradiction.
\end{proof}

An immediate consequence is the following.
\begin{cor}\label{cor:disj}
Let $\calC$ be a code over $\F_q$ with $q \neq 2$. Assume all
codewords of $\calC$ have even weight. Let $c_1,\ldots,c_r\in \calC$
of weight $2$ such that $c_i$ is not in $\langle c_j\rangle_{\F_q}$
for any $i\neq j$. Then
$$\supp(c_i)\cap \supp(c_j)=\emptyset,$$
for every $i\neq j$.
\end{cor}

So, we can get the first classification result.
\begin{lem} \label{lem:2roots}
  Let $\calC$ be a linear code of even length $n$ over $\F_q$ with $q \neq 2$.
  Suppose that
  \[
  W_\calC(x)=(x^2 + a)^{n/2}, \quad a \in \bbR\setminus\{0\}.
  \]
  Then $a = q-1$ and
  \[
  \calC \cong \bigoplus_{i=1}^{n/2}\langle(1,1)\rangle_{\F_q}.
  \]
\end{lem}

\begin{proof}
  If $n=2$, then clearly $\calC=\langle(1,1)\rangle_{\F_q}$.
  Write $W_\calC(x)= \sum_{i=0}^n a_i x^i$, so that $a_i$ is the number of codewords of weight $n-i$.
  Expanding the above expression, we see that $\calC$ has no codewords of odd weight.
  Moreover, the number of codewords of length $2$ is $a_{n-2} = an/2 \neq 0$.
  Let $r:= a_{n-2}/(q-1)$ and let $c_1,\ldots,c_r$ be a set of
  codewords of weight $2$ such that $c_i$ is not in $\langle c_j\rangle_{\F_q}$
  for any $i\neq j$. They have disjoint supports by
  Corollary \ref{cor:disj}.

  Let $S := \bigcup_i\supp c_i$ and let $\calC_S:= \langle c_1, \dotsc, c_r\rangle_{\F_q}$.
  Every codeword $x \in \calC$ can be written as a sum
  $$x = y + z,$$
  with $\supp(y) \subset S$ and $\supp (z) \cap S = \emptyset$. By
  Lemma \ref{lem:oddweight}, $y\in \calC_S\subset \calC$, so that $z\in
  \calC$.
  Consequently, $\calC$ is the direct sum
  \[
  \calC = \calC_S \oplus \calC_{S^c},
  \]
  where $\calC_{S^c}=\{c\in \calC \ | \ \supp(c)\cap S
  =\emptyset\}$.
  This implies that $W_\calC(x) = W_{\calC_S}(x) \cdot
  W_{\calC_{S^c}}(x)$.

  Now observe that $\calC_S$ is monomially equivalent to the code $\bigoplus_{i=1}^r \langle (1,1) 
  \rangle_{\F_q}$,
  and hence its weight enumerator is $W_{\calC_S}(x) = (x^{2} + (q-1))^r$.
  Therefore, we must have $a = (q-1)$. By induction, $\calC_{S^c}\cong \bigoplus_{i=1}^{n/2-r}
  \langle(1,1)\rangle_{\F_q}$ so that
  $\calC \cong \bigoplus_{i=1}^{n/2} \langle ( 1,1) \rangle_{\F_q}$, as desired.
\end{proof}

Let us now prove the classification theorem for codes whose weight
enumerator has at most two distinct roots in $\overline{\bbZ}$ (the
set of all algebraic integers).
\begin{thm}\label{thm:2roots}
Let $\calC \subset \F_q^n$ be a linear code with weight enumerator
$W_\calC(x) \in \bbZ[x]$ having at most two distinct roots in
$\overline{\bbZ}$. Then only the following possibilities may hold:
\begin{itemize}
\item[(a)] $W_\calC(x)=x^n$ and $\calC=\{\underline{0}\}$;
\item[(b)] $W_\calC(x)=(x+(q-1))^n$ and $\calC= \F_q^n$;
\item[(c)] $n$ is even, $W_\calC(x)=(x^2+(q-1))^{n/2}$ and, if $q\ne 2$,
$\calC\cong \bigoplus_{i=1}^{n/2}\langle(1,1)\rangle$.
\end{itemize}
\end{thm}

\begin{proof}
  Let $-a$, $-b$ be the roots of $W_\calC(x)$ in $\overline{\bbZ}$, so that
  $W_\calC(x) = (x + a)^r (x+ b)^{n-r}$ for $r \in \bbN$.
  The number of codewords in $\calC$ of weight one is then $c := ra +
  (n-r)b$.

  First, assume that $c \neq 0$ and let $m = c/(q-1)$. Taking arbitrary linear combinations of the codewords of weight one gives a copy of
  $\F_q^m$ in $\calC$. Therefore, $\calC = \calC_1 \oplus \calC_2$,
  with $\calC_1=\F_q^m$ and $\calC_2:=\{c\in \calC \ | \ \supp(c)\cap \supp(d)=\emptyset \ \forall d\in
  \calC_1\}$. Hence
  \[
  W_\calC(x) = (x + (q-1))^m \cdot W_{\calC_2}(x).
  \]
  Consequently, $-(q-1)$ is a root of $f$ and so either $a$ or $b$ is equal to $(q-1)$.
  Assume $a=(q-1)$ without loss of generality. We get
  \[
  m = \frac{ra + (n-r)b}{q-1} = r + \frac{(n-r)b}{q-1} \geq r.
  \]
  Hence, either $b=0$ and $r=n$, or $(x + (q-1))$ divides $(x+b)^{n-r}$,
  which implies $b=q-1$. Both cases give that
  $W_\calC(x) = (x + (q-1))^n$.
  But this implies $\#\calC = W_\calC(1) = q^n = \#\F_q^n$, whence $\calC = \F_q^n$, as
  desired.

  Now assume $\calC$ has no codewords of weight one, i.e. $c = ra+ (n-r)b =
  0$.
  If $a$ is real then so is $b$, and both must be non-negative: indeed, since $W_\calC(x)$ is non-zero and has
  positive coefficients, $W_\calC(y)>0$ for any positive real $y$, so $W_\calC$ has only non-positive roots.
  Since $ra = -(n-r)b$, we must have $a = b = 0$ whence $W_\calC$ has one root and $\calC =
  \{\underline{0}\}$.

  If $a$ is non-real, then $a$ $b$ are complex conjugate algebraic integers, and we must have
  $r = s$, which is possible only if $n$ is even.
  Consequently,
  \[
  W_\calC(x) = (x^2 + \Tr(a)x + \N(a))^{n/2},
  \]
  where $\Tr(a) = a + \bar a$ and $\N(a) = a\bar a$.
  The fact that $\calC$ has no codeword of weight one implies that $\Tr(a) = 0$, and hence,
  \[
  W_\calC(X) = (x^2 + \N(a))^{n/2}.
  \]
  
  If $q \neq 2$, Lemma \ref{lem:2roots} gives the desired conclusion about
  $\calC$.
  If $q = 2$, then we must show that $\N(a) = 1$. Since $a$ is an
  algebraic integer, $\N(a)\in \bbZ$. Since $q=2$, the number of
  codewords of weight $n$ is
  $$N(a)^{n/2}=1$$
  so that $N(a)=\pm 1$. If $N(a)=-1$ we have negative coefficients
  in $W_\calC(X)$, which is not possible.
\end{proof}

Note that Theorem \ref{thm:2roots} almost classify, up to monomial
equivalence, all linear codes with weight enumerator with at most
two distinct roots in $\overline{\bbZ}$. The case $q=2$ is left
unsolved and it seems quite difficult to be settled. If $q=2$, the
sum of two codewords of weight $2$ cannot have weight $3$, so that
the argument in the proof of Lemma \ref{lem:2roots} does not work.
\begin{que}\label{que:class}
Is it possible to classify all the binary codes of length $n$ with
weight enumerator $(x^2+1)^{n/2}$?
\end{que}

Let $\calC$ and $\calC'$ two codes with weight enumerator
$(x^2+1)^{n/2}$ and $(x^2+1)^{n'/2}$ respectively. Then
$\calC\oplus\calC'$ has weight enumerator $(x^2+1)^{(n+n')/2}$.
Hence, if we denote $$\mathcal{M}:=\{\text{binary codes of length} \
n \ \text{and weight enumerator} \ (x^2+1)^{n/2} \ | \ n\in
2\bbN\},$$ we have that $(\mathcal{M},\oplus)$ is a semigroup: in
order to answer positively to Question \ref{que:class} it suffices
to find all irreducible elements in $\mathcal{M}$, which means to
find a minimal set of generators of $(\mathcal{M},\oplus)$. We
consider, as usual, elements in $\mathcal{M}$ as classes of codes up
to equivalence.

Clearly, the generator with minimum length is the $[2,1,2]$ code
$\mathcal{X}_1:=\langle(1,1)\rangle$. Furthermore, every element in
$\mathcal{M}$ is formally self-dual and all formally self-dual codes
up to length $16$ are classified in \cite{BH01}. From an analysis of
the tables in the paper, we have that, up to length $16$, there are
exactly $4$ other irreducible elements of $\mathcal{M}$, namely the
formally self-dual (but not self-dual) $[6,3,2]$ code
$\mathcal{X}_2$ with generator matrix
$$\left[\begin{smallmatrix}
1 & 0 & 0 & 1 & 1 & 1 \\
0 & 1 & 0 & 1 & 1 & 1 \\
0 & 0 & 1 & 1 & 1 & 1
\end{smallmatrix}\right]$$
and three $[14,7,2]$ codes, which we call
$\mathcal{X}_3,\mathcal{X}_4$ and $\mathcal{X}_5$, with generator
matrices $[I|X_3]$,$[I|X_4]$ and $[I|X_5]$ respectively, where
$$X_3:=\left[\begin{smallmatrix}
1&1&1&0&0&0&0\\
1&1&1&0&0&0&0\\
1&1&1&0&0&0&0\\
1&0&0&1&1&1&1\\
1&0&0&1&1&1&1\\
1&0&0&1&1&1&1\\
1&0&0&0&0&0&0
\end{smallmatrix}\right], \quad
X_4:= \left[\begin{smallmatrix}
1&1&1&1&1&0&0\\
1&1&1&1&1&0&0\\
1&1&1&1&1&0&0\\
1&1&1&1&1&0&0\\
1&1&1&1&0&1&0\\
1&1&1&1&0&1&0\\
1&1&1&1&1&1&1
\end{smallmatrix}\right], \quad
X_5:= \left[\begin{smallmatrix}
1&0&1&0&1&0&0\\
1&0&1&0&1&0&0\\
1&0&1&0&1&0&0\\
1&0&1&0&1&0&0\\
1&1&1&0&1&0&1\\
1&1&1&0&1&0&1\\
1&1&1&1&1&1&1
\end{smallmatrix}\right]$$
and $I$ is the $7\times 7$ identity matrix. It is not clear how to
construct other generators and it seems already too complex for
software like {\sc Magma} \cite{Magma}. It is not evident if there
are infinitely many such generators or not.


We conclude this section showing a relation between our result and
the Gleason-Pierce Theorem (cf. \cite{Slo79} for a proof). Recall
that a code is divisible if there exists an integer $\Delta>1$ such
that the weight of every codeword of $\calC$ is divisible by
$\Delta$.

\begin{thm}[Gleason-Pierce]
Let $\calC$ be a formally self-dual code divisible by some $\Delta > 1$. Then
\begin{itemize}
\item $q=2$ and $\Delta\in \{2,4\}$,
\item $q=3$ and $\Delta=3$,
\item $q=4$ and $\Delta=2$,
\item or $q$ arbitrary, $\Delta=2$ and
$W_\calC(x)=(x^2+(q-1))^{n/2}$.
\end{itemize}
\end{thm}

Hence, Theorem \ref{thm:2roots} implies the following.

\begin{cor}
For $q>4$, if $\calC$ is a formally self-dual divisible code of
length $2n$, then $\calC$ is equivalent to the direct sum of $n$
copies of $\langle(1,1)\rangle_{\F_q}$.
\end{cor}


\section{Proof of Theorem \ref{thm-1}}\label{sec-proof1}

Let us start from a general lemma.

\begin{lem}\label{lemma-finite}
  Let $p(x,y) \in \bbC[x,y]$ be a homogeneous polynomial and let $n$ be the number of distinct complex roots of
  $p(x,1)$. If $n\geq 3$, then
  Then $$\#\Stab_{\GL_2(\bbC)}(p(x,y)) \leq n!\, \deg(p(x,1)).$$
  In particular, $\Stab_{\GL_2(\bbC)}(p(x,y))$ is finite if $n\geq 3$.
\end{lem}
\begin{proof}
  Let $V$ denote the projective variety defined by the vanishing of $p(x,y)$.
  Then $$V(\bbC) = \{(x:1) \in \bbP^1(\bbC) \: | \: f(x) = 0\},$$ so that $\#V(\bbC) \geq
  3$.
  Set $G := \Stab_{\GL_2(\bbC)}(p(x,y))$. Since every $A \in G$ fixes $p(x,y)$, $G$ acts on
  $V(\bbC)$, and since scalar matrices fix $V(\bbC)$ point-wise, this action induces an action of $\bar G \subset \PGL_2(\bbC)$ on $V(\bbC)$,
  where $\bar G$ is the image of $G$ in $\PGL_2(\bbC)$. It is
  well-known that $\PGL_2(\bbC)$ acts simply 3-transitively on
  $\bbP^1(\bbC)$. Since $\#V(\bbC) \geq 3$,
  every permutation of $V(\bbC)$ is realized by at most one $\bar A$ in $\PGL_2(\bbC)$, and hence by
  at most one $\bar A \in \bar G$.
  
  Such a $\bar A$ has exactly $\deg(p(x,1))$ pre-images in $G$: if $A\in G$ is such a pre-image,
  then all the pre-images in $\PGL_2(\bbC)$ are given by $\lambda A$ for 
  $\lambda\in \bbC\setminus\{0\}$.
  But $$(\lambda A)\cdot p(x,y) = \lambda^{\deg(p(x,1))}(A\cdot p(x,y))= \lambda^{\deg(p(x,1))}p(x,y),$$ so 
  that $\lambda A \in G$ if and only if $\lambda^{\deg(p(x,1))} =1$.
  There are $n!$ permutations of $V(\bbC)$, so at most $n!$ elements in $\bar
  G$.
  It follows that $\#G \leq n! \, \deg(p(x,1))$.
\end{proof}

Let us prove Theorem \ref{thm-1}, which we restate here for reader
convenience.
\begin{thm} The homogeneous weight enumerator of linear code has a finite number of
invariants if and only if its non-homogeneous version has at least
$3$ distinct roots in $\overline{\mathbb{Z}}$.
\end{thm}

\begin{proof}
Lemma \ref{lemma-finite} gives immediately that if the weight
enumerator has at least $3$ distinct roots in
$\overline{\mathbb{Z}}$, then the stabilizer of its homogeneous
version is finite.
If the weight enumerator of a linear code $\calC$ has at most $2$
distinct roots, then Theorem \ref{thm:2roots} implies that there exist
$a,m\in \bbN\setminus\{0\}$ such that $W_{\calC}(x,y)$ is equal to
one of the following:
$$x^m, \quad (x+ay)^m,\quad (x^2+ay^2)^{m}.$$
It is easy to observe that in all three cases we have an infinite
stabilizer.
\end{proof}


\section{The algorithm}\label{sec-algo}

The proof of Lemma \ref{lemma-finite} gives an algorithm to find the
stabilizer of every weight enumerator.

Let $\calC$ be a linear code. Suppose that its homogeneous weight
enumerator $W_\calC(x,y)$ is known and of degree $n$.
\begin{itemize}
\item[1.] Set $G:=\emptyset$.
\item[2.] Calculate $V(\bbC) := \{z_1, \dotsc,z_n\}$ the set of roots of
$W_\calC(x,1)$.
\item[3.] Call $V(\bbC)_3$ the set of all ordered 3-subsets of
$V(\bbC)$; we have $\#V(\bbC)_3=\frac{1}{6}n^3 - \frac{1}{2}n^2 +
\frac{1}{3}n$.
\item[4.] For every triple $\{w_1,w_2,w_3\}\in V(\bbC)_3$:
\begin{itemize}
\item[4a.] solve the system $z_ia+b-w_iz_ic-w_i d=0$, $i\in \{1,2,3\}$, where the unknowns are $a,b,c,d$.
It has clearly infinitely many solutions depending of one complex
parameter $\lambda$ (the action of $\PGL_2(\bbC)$ is simply
3-transitive, as we said). Call
$\underline{a},\underline{b},\underline{c},\underline{d}$ one
solution.
\item[4b.] If $\{\frac{\underline{a}z_i+\underline{b}}{\underline{c}z_i+\underline{d}} \ | \ z_i\in
V(\bbC)\}=V(\bbC)$, then
\begin{itemize}
\item[4bi.] let
$A:=\left[\begin{smallmatrix}\underline{a}&\underline{b}\\\underline{c}&\underline{d}\end{smallmatrix}\right]$.
\item[4bii.] Calculate $\lambda:=\frac{W_\calC(\underline{b},\underline{d})}{W_\calC(0,1)}$.
\item[4biii.] Let $G:=G\cup\{\zeta_{n}\lambda^{1/n}A\ | \ \zeta_{n}\in \bbC \ \text{s.t.} \ \zeta_n^n=1\}$.
\end{itemize}
\end{itemize}
\end{itemize}

Then $G$ is equal to $\Stab_{\GL_2(\bbC)}(W_\calC(x,y))$.\\

This algorithm can be implemented easily in {\sc Magma}, but there
is a problem for Step 2.: in $\bbC$, we do not have access to the
exact roots but only to approximations.
Of course, one can consider the splitting field of $W_\calC(x,1)$
instead of $\bbC$, but this does not give explicit solutions. So,
whenever Step 2. is feasible explicitly, the algorithm gives the
exact form for the stabilizer. Otherwise, we get an approximate
version.

To show that some Reed-Muller codes have weight enumerator with
trivial stabilizer, one needs to control the error made by approximating
the roots over $\bbC$.
The following lemmas are useful for this.

Recall that the cross ratio of four points $(z_1: 1), \dotsc, (z_4 :
1)$ is defined as
\[
[z_1, z_2, z_3, z_4] := \frac{(z_1- z_3)(z_2 - z_4)}{(z_1 - z_4)(z_2
- z_3)}.
\]
Make the symmetric group $S_4$ acts on the cross ratios by permuting
the points, and observe that for any $\sigma \in V_4:= \{\id,
(12)(34), (13)(24), (14) (23)\}$, we have $[z_{\sigma(1)},
z_{\sigma(2)}, z_{\sigma(3)}, z_{\sigma(4)}] = [z_1, z_2, z_3,
z_4]$.

Let $\mathcal{Z}$ be a set of at least four complex points. A
4-tuple of distinct elements $$\zeta = (z_1, z_2, z_3, z_4) \in
\mathcal{Z}^4$$ is called \emph{critical} if for any 4-tuple of
distinct elements $(y_1, y_2, y_3, y_4) \in \mathcal{Z}^4$, we have
\[
[z_1, z_2, z_3, z_4] = [y_1, y_2, y_3, y_4]
\]
if and only if $(y_1, y_2, y_3, y_4)= \zeta^\sigma= (z_{\sigma(1)},
z_{\sigma(2)}, z_{\sigma(3)}, z_{\sigma(4)})$ for some $\sigma\in
V_4$.

The following lemma gives the crucial argument to show that the
stabilizer is trivial.

\begin{lem}\label{lem-trivial}
Let $p(x,y)\in \bbC[x,y]$ be a polynomial with $5$
roots $z_1, z_2, z_3, z_4, z_5$ of $p(x,1)$ such that both $(z_1,
z_2, z_3, z_4)$ and $(z_1, z_2, z_3, z_5)$ are critical. Then
$\Stab_{\GL_2(\bbC)}(p(x,y))$ is trivial.
\end{lem}

\begin{proof}
Every $\bar A$ in $\PGL_2(\bbC)$ fixes $\{z_1, z_2, z_3, z_j\}$, for
$j=4,5$, since it must preserve the cross ratio of these four
points.
  If $\bar A$ sends $z_4$ to $z_j$ for $j \in \{1,2,3\}$, then from the fact that $\bar A$ fixes
  $\{z_1, z_2, z_3, z_5\}$ if follows that some element of this set is also sent to $z_j$,
  contradicting the injectivity of $\bar A$.
  Thus $\bar A z_4 = z_4$.
  But the only permutation of $(z_1, z_2, z_3, z_4)$ which fixes the cross ratio and sends $z_4$ to
  $z_4$ is the identity.
  Thus $\bar A$ fixes four points of $\bbP^1(\bbC)$.
  Since the action is sharply 3-transitive, $\bar A = \id$, and the conclusion follows.
\end{proof}

In order to prove the triviality of the stabilizer, Lemma \ref{lem-trivial} implies that it suffices
to find 5 roots with the above property. This is feasible even with
approximated roots, provided that we can control the errors.

\begin{lem}\label{lem:tech}
  Let $z_1,z_2,z_3,z_4$ and $\delta_1,\delta_2,\delta_3,\delta_4$ be complex numbers.
  Let $M = \max_i|z_i|$ and $\delta = \max_i |\delta_i|$.
  Assume $\delta < 1$.
  Then
\[
  |(z_1 + \delta_1)(z_2 + \delta_2)(z_3 + \delta_3)(z_4 + \delta_4) - z_1z_2z_3z_4| \leq 15 M^3 \delta.
  \]
\end{lem}
\begin{proof}
  Expand the expression.
  The term $z_1z_2z_3z_4$ cancels out, so we are left with $15$ terms, all of which contain at least
  one of the $\delta_i$, and at most three $z_i$'s.
  Explicitly, all remaining terms are of the form
  \[
  \prod z_i \cdot \prod \delta_j
  \]
  where the first product ranges over at most 3 indices $i$, and the second ranges over at least
  1 index $j$.
  Since $\delta < 1$, all these terms are less than $M^3\delta$ in module, whence the conclusion.
\end{proof}

\begin{cor} \label{cor:approx}
  Let $x_1, \dotsc, x_8$ be 8 complex numbers, and let $\tilde x_1, \dotsc, \tilde x_8$ be
  approximations so that
  \[
  \max_{j} |x_j - \tilde x_j| \leq \epsilon < 1/2
  \]
  for a fixed $\epsilon$.
  Let $N := \max_j | x_j|$.
  Then:
  \begin{enumerate}
  \item \label{it:1} $\max_{i,j} |(x_i - x_j) - (\tilde x_i - \tilde x_j)| \leq 2 \epsilon$
  \item \label{it:2} $
    |(x_1 - x_3)(x_2 - x_4)(x_5 - x_8)(x_6 - x_7) -
    (\tilde x_1 - \tilde x_3)(\tilde x_2 - \tilde x_4)(\tilde x_5 - \tilde x_8)
    (\tilde x_6 - \tilde x_7) | \leq 60 N^3 \epsilon
    $.
  \item \label{it:3} Let $a = (x_1 - x_3)(x_2 - x_4)(x_5 - x_8)(x_6 - x_7)$,
    $b = (x_1 - x_4)(x_2 - x_3)(x_5 - x_7)(x_6 - x_8)$, and let $\tilde a$, $\tilde b$ denote the
    same expressions where the $x$'s have been replaced by their approximations.
    If $|\tilde a - \tilde b| > 120 N^3 \epsilon$, then $|a - b| > 0$.
  \end{enumerate}
\end{cor}
\begin{proof}
  Part (\ref{it:1}) is just the triangle inequality.
  Part (\ref{it:2}) follows from Lemma \ref{lem:tech} applied to the various differences $(x_i - x_j)$ in place of
  the
  $a$'s and the approximations $(x_i - x_j) - (\tilde x_i - \tilde x_j)$ in place of the $\alpha$'s,
  and noting that the $M$ in Lemma \ref{lem:tech} is less than twice $N$, and
  that the maximal module of the $\alpha$'s is less than $2 \epsilon$.

  Finally, for Part (\ref{it:3}) it suffices to observe that
  \[
  |\tilde a - \tilde b| \leq |\tilde a - a| + |a - b| + |b - \tilde b|
  \iff
  |\tilde a - \tilde b| - |\tilde a - a| - |b - \tilde b| \leq  |a - b| .
  \]
  Now the previous part applied to different permutations of $x_1, \dotsc, x_8$ certainly imply that
  \[
  |\tilde a - a|, |b - \tilde b| \leq  60 N^3 \epsilon.
  \]
\end{proof}


\section{Invariants of some Reed Muller codes}\label{sec-triv}

Using this algorithm, we found some non-trivial invariants of
particular Reed-Muller codes. This is not useful for determining
unknown weight enumerators, but rather shows an application of our
procedure.

\begin{pro}
  Let $\zeta_{2^m}$ be a $2^{m}$-root of unity, and let $u_m := \frac{\zeta_{2^m} + 1}{2}$.
  For any $m \geq 3$,
    \[
    \begin{bmatrix}
      u_m & u_m-1 \\
      u_m-1 & u_m
    \end{bmatrix}
    \in \Stab_{\GL_2(\bbC)}(W_{\RM_2(m-2, m)}(x,y)).
    \]
\end{pro}
\begin{proof}
The code $\RM_2(m-2, m)$ is the dual of $\RM_2(1, m)$, the code coming from the evaluation of linear
polynomials over $\F_2$ (Theorem 4, Chapter 13 \cite{MS}).
It is easy to see that the weight enumerator of $\RM_2(1,m)$ is
\[
W_{\RM_2(1, m)}(x,y) = X^{2^m} + 2(2^m-1)x^{2^{m-1}}y^{2^{m-1}} + y^{2^m}.
\]
Now by MacWilliams' Theorem,
\[
W_{\RM_2(m-2,m)} = \dfrac{1}{2^{m+1}} \left(
    (X-Y)^{2^m} + 2\cdot (2^{m} - 1)(X-Y)^{2^{m-1}} (X+Y)^{2^{m-1}} + (X+Y)^{2^m}
    \right).
\]
It is immediate to check that the above matrix fixes this
polynomial.
\end{proof}

On the other hand, Lemma \ref{lem-trivial} and Corollary \ref{cor:approx} give a way to show that some codes
have weight enumerators with trivial stabilizers.

\begin{thm}
  The following codes have weight enumerator with trivial stabilizer in $\PGL_2(\bbC)$, i.e. with
  stabilizer $\GL_2(\bbC)$ consisting only of scalar matrices:
  \[
  \RM_{4}(2,2), \quad \RM_{4}(3,2), \quad \RM_{4}(2,2),
  \]
  \[
  \PRM_{5}(3,2), \quad \PRM_{5}(3,2)^\perp = \PRM_{5}(5,2).
  \]
\end{thm}
\begin{proof}
  For every code $\calC$ in the list above, we can proceed as follows:
  it can be checked that $W_\calC(x,1)$ has at least 5 roots in $\bbC$.
  The goal is two find two critical $4$-tuples which differ only in one element, as in
  Lemma \ref{lem-trivial}.
  The problem lies in the fact that we can only work with approximations.
  However, a computer software like {\sc Magma} gives the roots of a polynomial up to arbitrarily
  precision, which will help us finding critical 4-tuples.

  Let $x_1, \dotsc, x_8$ be complex points, and let $a,b$ be as in part (\eqref{it:3}) of
  Corollary \ref{cor:approx}.
  Easy algebraic manipulations show that for $[x_1, x_2, x_3, x_4] \neq [x_5, x_6, x_7, x_8]$ if
  and only if $|a-b| >0$.
  Therefore, by choosing the approximations of the roots of $W_\calC(x,1)$ close enough, it is
  possible to get the $\epsilon$ of Corollary \ref{cor:approx} arbitrarily small.
  This gives us a way to find critical 4-tuples, as desired.
\end{proof}


\section*{Acknowledgements}
The results of this paper are partially in the Master thesis
\cite{BFBM} of the second author. The authors express their deep
gratitude to Eva Bayer Fluckiger and Peter Jossen for their support
and the fruitful discussions.


\bibliographystyle{alpha}

\begin{thebibliography}{HGRS11}

\bibitem[A64]{A} J.~Ax.
\newblock \emph{Zeroes of polynomials over finite fields}.
\newblock Amer. J. Math. 86: 255--261 (1964).

\bibitem[BCP97]{Magma} W.~Bosma, J.~Cannon and C.~Playoust.
\newblock \emph{The Magma algebra system I: The user language}.
\newblock J. Symbol. Comput. 24: 235--265 (1997).

\bibitem[Ken94]{Ken94}
G.~T.~Kennedy.
\newblock \emph{Weight distributions of linear codes and the {G}leason-{P}ierce theorem.}
\newblock J. Combin. Theory Ser. A 67(1): 72--88 (1994).

\bibitem[BH01]{BH01}
K.~Betsumiya and M.~Harada.
\newblock \emph{Classification of formally self-dual even codes of lengths up to 16}.
\newblock Des. Codes Cryptogr. 23(3): 325--332 (2001).

\bibitem[E06]{El} N. D. Elkies.
\newblock \emph{Linear codes and algebraic geometry in higher
dimensions}.
\newblock Preprint (2006).

\bibitem[G70]{Gl} A.M.~Gleason.
\newblock \emph{Weight polynomials of self-dual codes and the {M}ac{W}illiams
              identities}.
\newblock Actes du {C}ongr\`es {I}nternational des {M}ath\'ematiciens   ({N}ice, 1970), {T}ome 3: 211--215 (1971).

\bibitem[K13]{Na} N.~Kaplan.
\newblock \emph{Rational Point Counts for del Pezzo Surfaces over Finite
Fields and Coding Theory}.\\
\newblock \url{http://users.math.yale.edu/~nk354/papers/kaplanthesis.pdf} (2013).

\bibitem[MS77]{MS} F.~J.~MacWilliams and N.~J.~A.~Sloane.
\newblock \emph{The theory of error-correcting codes.}
\newblock I. North-Holland Publishing Co., Amsterdam-New York-Oxford. North-Holland Mathematical Library, Vol. 16 (1977).

\bibitem[M15]{BFBM} O. Mila.
\newblock \emph{Invariance
for weight enumerators of evaluation codes and counting
$\F_q$-rational points on hypersurfaces}.
\newblock \url{
http://csag.epfl.ch/files/content/sites/csag/files/MasterOlivierMila.pdf}
(2015).

\bibitem[Slo79]{Slo79}
N.J.A.~Sloane.
\newblock \emph{Self-dual codes and lattices}.
\newblock Proc. Symp. Pure Math. 34: 273--308 (1979).


\end{thebibliography}

\newpage
\footnotesize

\noindent MARTINO BORELLO
\vspace{1mm}\\
{\sc \'{E}cole Polytechnique F\'{e}d\'{e}rale de Lausanne\\
SB MathGeom CSAG\\
B\^{a}timent MA\\
Station 8\\
CH-1015 Lausanne\\
Switzerland}
\vspace{1mm}\\
E-mail address:
\href{mailto:martino.borello@epfl.ch}{martino.borello@epfl.ch}

\vspace{3mm}

\noindent OLIVIER MILA
\vspace{1mm}\\
{\sc Universit\"{a}t Bern\\
Mathematisches Institut (MAI)\\
Sidlerstrasse 5\\
CH-3012 Bern\\
Switzerland}
\vspace{1mm}\\
E-mail address:
\href{mailto:olivier.mila@math.unibe.ch}{olivier.mila@math.unibe.ch}

\end{document}